\documentclass{llncs}

\usepackage{amsmath,amssymb}
\usepackage{listings}
\usepackage{hyperref}
\usepackage{llncsdoc}
\usepackage{graphicx}
\usepackage{float}
\usepackage[ruled]{algorithm2e}
\usepackage{algpseudocode}
\usepackage{soul}
\usepackage{color}
\usepackage{tikz}
\usepackage{multirow}

\newtheorem{thm}{Theorem}

\newtheorem{lem}{Lemma}
\newtheorem{defi}{Definition}

\begin{document}

\title{Orbit Problem Revisited}

\author{Taolue Chen$^{1}$, Xiaoming Sun$^{2}$ and Nengkun Yu$^{3,4}$}
\institute{ $^{1}$Department of Computer Science, University of Oxford, UK \\
$^{2}$Institute of Computing Technology, Chinese Academy of Sciences,
Beijing, China \\$^{3}$State Key Laboratory of Intelligent Technology and Systems, Tsinghua National Laboratory
for Information Science and Technology, Department of Computer Science and Technology,
Tsinghua University, Beijing, China \\
$^{4}$Center for Quantum Computation and Intelligent Systems (QCIS), \\ Faculty of
Engineering and Information Technology, \\ University of Technology,
Sydney, NSW 2007, Australia}

\maketitle

\begin{abstract}
In this letter, we revisit the {\em orbit problem}, which was studied in \cite{HAR69,SHA79,KL86}. In \cite{KL86}, Kannan and Lipton proved that this problem is decidable in polynomial time. In this paper, we study the {\em approximate orbit problem}, and show that this problem is decidable except for one case.

\end{abstract}

\section{Introduction}
A sequential machine \cite{ELS59,GIL64} is called \emph{linear} if its ``next state" function and ``output" function are linear transformations.
Such machines have many important applications, for instance, computer control circuitry, automatic error-correction circuits, digital communication systems etc. It is very natural to study the accessibility problem for linear sequential machines, i.e., decide
whether there is an input $x$ such that on $x$ the machine, starting in a given state $q$, goes to a given state
$q'$. In \cite{HAR69}, Harrison proved that this problem can be reduced to a standard linear algebra problem, called the ``orbit problem".

\begin{defi}[Orbit problem] Given an $n\times n$ rational matrix $A$, a rational initial vector $\vec{x}$ and a rational target vector $\vec{y}$, decide whether there is some $k\in \mathbb{N}$ such that $A^k \vec{x}=\vec{y}$.
\end{defi}

Shank showed that this problem is decidable for $n=2$ \cite{SHA79}, while the general decidability result was obtained by Kannan and Lipton \cite{KL86}. They related this problem to another linear algebra problem, the \emph{matrix power problem}, which is the polynomial analog of the well-known discrete logarithm
problem for fields of prime characteristic \cite{ADE79}.

\begin{defi}[Matrix power problem]
Given two $n\times n$ rational matrices $A$ and $B$, decide whether there exists some $k\in \mathbb{N}$ such that $A^k=B$.
\end{defi}

Kannan and Lipton show that the ``matrix power problem" for general $n$ is decidable in polynomial time, so is the ``orbit problem". A natural generalization of the ``matrix power problem" is the $ABC$ problem studied by Cai $et.~al.$ \cite{CLZ94,CAI94}. They proved that it is decidable in polynomial time whether there exists natural numbers $k$ and $j$ such that $A^kB^j=C$ for given rational matrices $A,B,C$ with $AB=BA$.

\medskip

Some generalized versions of the orbit problem have been addressed in literature.
In this paper, we study the ``approximate orbit problem" which is formulated as follows.

\begin{defi}[Approximate orbit problem] Given a rational $n\times n$ matrix $A$, a rational initial vector $\vec{x}$, a rational target vector $\vec{y}$ and a rational number $\delta>0$, decide whether there exists some $k\in \mathbb{N}$ such that
\[||A^k \vec{x}-\vec{y}||<\delta,\]
where $||\cdot||$ is some given norm of the linear space.
\end{defi}

Throughout the paper, we use $V$ to denote the linear vector space; $e_i$ denotes the standard unit basis of $V$ and $M_n$ denotes all linear mappings on $V$. It is convenient to define a linear bijection
\[vec : M_n\mapsto V\otimes V
\]
by $vec(E_{ij})=e_i\otimes e_j$, where $E_{ij}$ is the square matrix with all zeros except a 1 at location $(i,j)$.

A square matrix $U$ is called unitary if $U^{\dag}U=UU^{\dag}=I$ with $I$ being the identity matrix. It is well known that any unitary $U$ has the spectral decomposition, $i.e.$, $U=V^{\dag}DV$ for some unitary $V$ and diagonal unitary $D$, $V^{\dag}$ stands for the complex conjugate of $V$.

$\mathcal{S}(A,\vec{x}):=\{\vec{x},A\vec{x},A^2\vec{x},\cdots,A^m\vec{x},\cdots\}$ is used to denote the orbit defined by a matrix $A$ and vector $\vec{x}$.

An algebraic number is a root of a polynomial in $\mathbb{Q}[x]$, the ring of polynomials
in the variable $x$ over the rationals.
We will associate an algebraic number with an irreducible polynomial with rational coefficients and a sufficiently good rational approximation, which uniquely identify the particular root of the polynomial. A root of unity is any complex number that equals 1 when raised to some integer power $n$.

\paragraph{Structure.} In order to study the ``approximate orbit problem". Firstly, we gives a fully characterization of the set of the limit points of the given orbit. After that, we provide a method to solve the ``approximate orbit problem" when $\epsilon$ does not equal to the distance between the target vector and the set of the limit points, which could be well approximated using epsilon net. We conclude the paper in section 4.

\section{Approximate Orbit Problem}
In this section, we study one generalization of the orbit problem, the ``approximate orbit problem", that whether the given orbit could arrived in some given open ball. We show that this problem is decidable except for one case.

Assume an orbit in a linear space with norm $||\cdot||$,  \[\mathcal{S}(A,\vec{x})=\{\vec{x},A\vec{x},A^2\vec{x},\cdots A^m \vec{x}, \cdots\},\]
a target vector $\vec{y}$ and a radius $\delta$. The problem asks whether there exists some $k\in \mathbb{N}$ such that $A^k\vec{x}\in B(\vec{y},\delta)$ with $B(\vec{x},\delta)=\{\vec{z}:||\vec{x}-\vec{z}||<\delta\}$ being the open ball of radius $\delta$ and centered at $\vec{y}$.

We introduce the concept of  limit points of $\mathcal{S}(A,\vec{x})$.

\begin{defi}
 $\vec{p}$ is a \emph{limit point} if for any $\epsilon>0$, there exist at least two different integers $k_1\neq k_2\in\mathbb{N}$ such that, the open ball $B(\vec{p},\epsilon)$ contains $A^{k_1}\vec{x}$ and $A^{k_2}\vec{x}$, $i.e.$,
\[||\vec{p}-A^{k_1}\vec{x}||,||\vec{p}-A^{k_2}\vec{x}||<\epsilon.\]
We write $S_L$ for the set of limit points.
\end{defi}
Assume that $A$, $\vec{x}$,$\vec{y}$ and $\delta$ are all rational. 
Our analysis precedes by three steps:

\begin{itemize}
  \item Step 1, we give a parameterized characterization of $S_L$.

  \item Step 2, we check whether $S_L$ is empty. It is not hard to see that if $S_L=\varnothing$, the ``approximate orbit problem" is easy to solve.
Otherwise, more precise analysis is needed: we approximate the distance between $\vec{y}$ and $S_L$ as well as possible, where the distance defined as following
\[D(\vec{y}, S_L)=\inf_{\vec{q}\in S_L} D(\vec{y},\vec{q})=\inf_{\vec{q}\in S_L} ||\vec{y}-\vec{q}||.\]

  \item Step 3: 
      Assuming that $\delta\neq D(\vec{y},S_L)$. We separate the following two cases 
      by approximating $D(\vec{y},S_L)$.
      \begin{itemize}
        \item $\delta>D(\vec{y},S_L)$. One can conclude that there exists some $k\in \mathbb{N}$ such that
\[D(A^k \vec{x},\vec{y})<\delta,\]
in fact, there are infinitely many such $k$.

\item$\delta<D(\vec{y},S_L)$. Namely, any limit point of $\mathcal{S}$ does not lie in the open ball of radius $\delta$ centered at $\vec{y}$. To handle this case, we need analysis the distance between the limit point of $S$ and non-limit point of $\mathcal{S}$. Then, we bound the possible $k$ such that $A^k \vec{x}\in B(\vec{y},\delta)$.
      \end{itemize}

\end{itemize}

\subsection{Characterization of $S_L$}
In this subsection, we give the full parameterized characterization of $S_L$, the set of limit points of the orbit $\mathcal{S}$.

Before studying the general case, we note the following well known lemma describes the closure of the matrix group generated by one diagonal unitary, where the closure is defined according to the matrix norm induced by the given vector norm $||
\cdot||$.
\begin{lem}\label{lem:uniatry}
$U=diag(u_1,u_2,\cdots,u_n)$ is a diagonal unitary with $\frac{u_i}{u_j}$ and $u_i$ being no root of unity for any $i\neq j$. Then, the closure of $\{U,U^2,\cdots,U^m,\cdots\}$ is the set of all diagonal unitaries.
\end{lem}

\begin{remark}
Apply this lemma on orbit $\mathcal{S}=(A,\vec{x})$ with $A$ being such diagonal unitary, we know that $S_L=\{V\vec{x}:V~\mathrm{ranges~~over~~all~~diagonal~~unitaries}~\}$.
\end{remark}

\begin{remark}
It is direct to see that any unitary $U$ lies in the set of the limit points of the group generated by itself. By employing the spectral decomposition of $U$, we only need to deal with the diagonal case, then divide all the eigenvalues into disjoint set, and apply the above lemma.
\end{remark}

In general, we have the following result:
\begin{lem}\label{lem:limitpoint}
For given $\mathcal{S}(A,\vec{x})=\{\vec{x},A\vec{x},A^2\vec{x},\cdots A^m \vec{x}, \cdots\}$, we have an algorithm which classifies $S_L$ into the following cases,
 \begin{itemize}
\item Case A: $S_L$ is empty, and computes a integer $N$ and a constant $c>0$ such that and for $m>N$, one have $||A^m \vec{x}||>c m$.

\item Case B: $S_L$ is not empty, and computes a non-singular matrix $Q$ and a diagonal matrix $C$ whose diagonal entries are zero or with absolute value 1 and $\vec{v}$ such that $S_L$ equals to the closure of $Q\{C\vec{v},\cdots,C^m\vec{v},\cdots\}$.
 \end{itemize}
\end{lem}
\begin{proof} Suppose the Jordan decomposition of $A$ is given by $A=PJP^{-1}$ with
$P$ being a nonsingular matrix, and $J$ being the Jordan normal
form of $A$:
\[J= diag (J_{k_1}(\lambda_1),
J_{k_2}(\lambda_2),\cdot\cdot\cdot, J_{k_l}(\lambda_l)),\] where
$J_{k_s}(\lambda_s)$ is a $k_s\times k_s$-Jordan block with corresponding eigenvalue $\lambda_s$ $(1\leq s\leq l)$.

Let $\vec{z}=P^{-1}\vec{x}$. We have that $\mathcal{S}(A,\vec{x})=\{P \vec{z}, P J\vec{z},\cdots,PJ^m\vec{z},\cdots\}$. As $P$ is non-singular, $S_L$ can be obtained by applying $P$ on the limit points of \[\mathcal{S}(J,\vec{z})=\{\vec{z}, J\vec{z},\cdots,J^m\vec{z},\cdots\}.\]

To consider the limit points of $\mathcal{S}'$, we first deal with the following orbit with $J_{t}(\lambda)$ being the $t\times t$-Jordan block of
eigenvalue $\lambda$ and $\vec{r}\neq 0$.
\[\mathcal{S}(J_{t}(\lambda),\vec{r})=\{\vec{r},J_{t}(\lambda)\vec{r},\cdots,J_{t}^m(\lambda)\vec{r},\cdots\}.\]
There are four cases,
 \begin{itemize}
\item Case i). $|\lambda|>1$. Then $\mathcal{S}(J_{t}(\lambda)\vec{r})$ has no finite limit point, as there exist some constant $a>1$ and integer $N$ such that for any $m>N$,
\[||J_{t}^m(\lambda)\vec{r}||>a^m\]

\item Case ii). $|\lambda|<1$. Then $\mathcal{S}(J_{t}(\lambda),\vec{r})$ has only one limit points $\vec{0}$ as there exist some constant $0<b<1$ and integer $N'$ such that for any $m>N'$,
    \[ \lim_{m\rightarrow \infty }||J_{t}^m(\lambda)\vec{r}||<b^m.
    \]

\item Case iii). $|\lambda|=1$ and $t=1$. Then $J_{t}^m(\lambda)\vec{r}=\lambda^m r_{1}$.

\item Case iv). $|\lambda|=1$ and $t>1$. Then if $\vec{r}$ has at least two nonzero entries or the only nonzero entry does not lie in the last entry, the points of the set will goes to infinity. there exist some constant $c>1$ and integer $N$ such that for any $m>N$,
     \[||J_{t}^m(\lambda)\vec{r}||>cm.\]
    Otherwise, for any $m$, the only non-zero entry of $\vec{r}$ is the first entry, $r_1$. Then the first entry of the $J_{t}^m(\lambda)\vec{r}$ is $\lambda^m r_{1}$.
 \end{itemize}
Applying the above result, we are able to check whether $\mathcal{S}(J,\vec{z})$ is empty or not.

If it is $S_L$ is empty, then there are two cases:
 \begin{itemize}
 \item Case 1. According to case i), we know that there is some $a>1$ and a integer $N_1$ such that $$||J^m\vec{z}||> a^m$$ is true for any $m>N_1$;
 \item Case 2. According to case iv), we know that there is some $c>0$ and a integer $N_2$ such that $$||J^m\vec{z}||> c m$$ is true for any $m>N_2$.
  \end{itemize}
 Therefore, for $m>\max\{N_1,N_2\}$, we have $$||J^m\vec{z}||> c m ~\Rightarrow~||PJ^m\vec{z}||> c'm~\Rightarrow~||A^m\vec{x}||> c' m$$ for $m>\max\{N_1,N_2\}$, where $c'=\frac{c}{||P^{-1}||}$ as $||P^{-1}||$ being the induced norm of $P^{-1}$.

\item Case B: $S_L$ is not empty, invoking observation ii), iii) and iv), we know that $S_L$ equals to the set of limit points of
\[\{P\vec{z}, PC\vec{z},\cdots,PC^m\vec{z},\cdots\}=P\{\vec{z},C\vec{z},\cdots,C^m\vec{z},\cdots\},\]
    where $C$ is obtained by modifying $J$ according to $\vec{z}$ as follows
    \[C=\begin{cases} 0 &{\rm if}\ |\lambda_s|\neq 1,\\ J'_{k_s}(\lambda_s)
&{\rm otherwise,}\end{cases}\] for each $1\leq s\leq l$, where $J'_{k_s}(\lambda_s)$ is obtained by replacing any element by zero, but the only element lies in the $(1,1)$ position.

Now, $C$ is the wanted diagonal matrix whose non-zero diagonal entries are absolution 1 and $Q=P$. That is, the set of limit point of $\mathcal{S}$ is the closure of $Q\{\vec{v},C\vec{v},\cdots,C^m\vec{v}\}$.

In order to see that, we only need to notice that for any $m\geq 1$, $C^m\vec{v}$ is a limit point of $\{\vec{v},C\vec{v},\cdots,C^m\vec{v}\}$, where $\vec{v}=\vec{z}$. Thus, the limit point of $\{C\vec{v},\cdots,C^m\vec{v}\}$ is the closure of itself.
\qed
\end{proof}

Combining Lemma~\ref{lem:uniatry} and Lemma~\ref{lem:limitpoint}, we give the parameter description of $S_L$ as follows. Firstly, we check whether $S_L=\varnothing$.
 If this is not the case, we deal with the orbit
$\{\vec{v},C\vec{v},\cdots,C^m\vec{v},\cdots\}$.  We consider the equivalence $\equiv$ over the non-zero eigenvalues of $C$: two eigenvalues $\lambda_1\equiv\lambda_2$  if $\lambda_1/\lambda_2$ is a root of unity. $\equiv$ gives rise to a partition $\{S_i\}$.
%
Without loss of generality, we assume that
\begin{eqnarray*}
C=\left(
\begin{array}{cc}
D & 0  \\
 0 & 0 \\
\end{array}
\right),
\end{eqnarray*}
where $D=\oplus_{j} D_j$, with $D_j$ being a diagonal unitary matrix whose eigenvalues are in $S_j$.

There is an integer $N_j$ such that $D_j^{N_j}$ equals to $I$ up to some global phase. Thus, there is a universal $N$ such that $D_j^{N}$ equals to $I$ up to some global phase. Now we divide $\{\vec{v},C\vec{v},\cdots,C^m\vec{v},\cdots\}$ into $N$ orbits, the $k$-th orbit is \[\mathcal{S}(C^N,C^k\vec{v})=\{C^k\vec{v},C^{N+k}\vec{v},\cdots,C^{N+mk}\vec{v},\cdots\}.\]
If there is no $j$ such that the eigenvalues of $S_j$ are unit root, that is $D_j^l\neq I$ for any $l$,
according to Lemma 5, the set of limit points of $\mathcal{S}(C^N,C^k\vec{v})$ is
\[S_{L}^k=\{\oplus_{j} e^{i\alpha_j}D_j^{k}\cdot\vec{v}:\alpha_j\in \mathcal{R}\}.\]
If there is $j$ such that the eigenvalues of $S_j$ are unit root, says $S_1$,
according to Lemma 5, the set of limit points of $\mathcal{S}(C^N,C^k\vec{v})$ is
\[S_{L}^k=\{D_1^{k}\oplus_{j>1} e^{i\alpha_j}D_j^{k}\cdot\vec{v}:\alpha_j\in \mathcal{R}\}.\]
Thus, we can obtain \[S_L=\bigcup Q\cdot S_{L}^k,\] the union of finite parameter described boundary closed sets, a compact set.

\subsection{Approximate the distance between $\vec{y}$ and $S_L$}

In this subsection, we give a simple algorithm to approximate the distance between $\vec{y}$ and $S_L$.

For convenience, we assume that no there is no $j$ such that the eigenvalues of $S_j$ are unit root. And the idea can be directly used to approximate the distance for that case with $S_1$ contains unit root eigenvalues.

Firstly, we provide an algorithm which could well approximate the distance between $\vec{y}$ and $S_{L}^k=\{Q\oplus_{j=1}^h e^{i\alpha_j}D_j^{k}\cdot\vec{v}:\alpha_j\in \mathcal{R}\}$, that is
\[ d_k:=\inf f_k(\mu_1,\mu_2,\cdots,\mu_h),
\]
with restrictions that $|\mu_j|=1$ for any $j$, in other words, $d_k$ is the infimum of $f_k$ over the direct product of $h$ unit circle. The function $f_k$ is defined on the direct product of $h$ unit circle
\[f_k(\mu_1,\mu_2,\cdots,\mu_h)=||\vec{y}-Q\oplus \mu_kD_j^{k}\cdot\vec{v}||.\]
The function $f_k$ is a continuous function (actually it is a Lipschitz function) since the norm function $||\cdot||$ is continuous, and the domain is a compact set, the minimum value is achievable.

We will use epsilon-net to give a well approximation of $d_k$: For any given $\epsilon>0$, we can choose a set $X$ of finite number of points in the feasible set such that for any point of the domain, there is a point of $X_\epsilon$ such that the distance of this two points is less than $\epsilon$. We choose $d_k(X_\epsilon)$ as the minimal value of $f_k$ over the finite points set $X_\epsilon$ as a approximation of $d_k$, that is,
\[d_k(X_\epsilon):=\min_{(\nu_1,\nu_2,\cdots,\nu_h)\in X_\epsilon}f_k(\nu_1,\nu_2,\cdots,\nu_h).\]
Since the function $f_k$ is Lipschitz, we know that $d_k(X_\epsilon)$ is close to $d_k$. To get better approximation of $d_k$, we only need to choose smaller $\epsilon$..

For any $k$, we can have well approximation of $d_k$, thus, $d(X_\epsilon)=\min_{k}d_k(X_\epsilon)$ could be made sufficient close to $D(\vec{y},S_L)$. More precisely, according to the Lipschitz property, for any $\xi>0$, we can choose $\epsilon$ such that for any epsilon net $X_\epsilon$ of the direct product of $h$ unit circle, we have
\[d(X_\epsilon)<D(\vec{y},S_L)+\xi.\]

\subsection{Solve the Approximate Orbit Problem}
In this subsection, we attempt to solve the ``approximate orbit problem".
At the first step, we check whether $S_L$ is an empty set.
\begin{thm}
If $S_L=\varnothing$, the ``approximate orbit problem" is decidable.
\end{thm}
\begin{proof} $S_L$ is empty, we can find $K$  and $c>0$ such that for $k>K$ \[||D(y,A^k \vec{x})||=||y-A^k\vec{x}||\geq ||A^k\vec{y}||-||y||> ck-||y||.\]
Thus, we only need to verify whether there is some $k\leq\max\{\frac{\epsilon+||y||}{c},K\}$ that whether $||D(y,A^k \vec{x})||\leq \epsilon$ is true in this case.
\qed
\end{proof}
\begin{thm}
If $S_L\neq \varnothing$, the ``approximate orbit problem" is decidable provided that $\epsilon\neq D(\vec{y},S_L)$.
\end{thm}
\begin{proof} Firstly, we check which the following cases it is: $\epsilon>D(y,S_L)$ or $\epsilon<D(y,S_L)$ by making good enough approximation of $D(y,S_L)$.

According to the previous section, we can construct a decreasing series $x_j$ satisfying $0<x_j-D(y,S_L)<2^{-j}$ holds for any $j$. Now we have two new series: $a_j=x_j-\epsilon$ and $b_j=x_j-\epsilon-2^{-j}$ for $j\geq 0$. It is obviously that \[a_j\geq D(y,S_L)-\epsilon=d_j-\epsilon-(d_j-D(y,S_L))>d_j-\epsilon-2^{-j}=b_j.\]
Moreover, $a_j$ and $b_j$ share the limit $D(y,S_L)-\epsilon$. If $\epsilon>D(y,S_L)$, we know that there exists $j$ such that $a_j<0$. If $\epsilon<D(y,S_L)$, there exists $j$ such that $b_j>0$. The idea to distinguish these two cases is: Check whether $a_j<0$, if the answer is yes, $\epsilon>D(y,S_L)$; otherwise, check whether $b_j>0$, if the answer is yes, $\epsilon<D(y,S_L)$. When the two answers are both no, Check the same thing for $j+1$. We can always distinguish the two cases, provided $\epsilon\neq D(y,S_L)$,
We summarize this algorithm below.
\smallskip\
\begin{algorithm}
\caption{Bound $\epsilon-D(y,S_L)$ \label{alg:AT}}
\SetKwInOut{Input}{input}\SetKwInOut{Output}{output}
\Input{integer $n$}
\Output{nonzero $\eta$. (a number between 0 and $\epsilon-D(y,S_L)$)}
\textbf{real number} $b\leftarrow 0$\;
\textbf{integer} $j\leftarrow 1$\;
\While{$j\neq 0$}{
Construct $x_j$ satisfying $0<x_j-D(y,S_L)<2^{-j}$\;
\textbf{real number} $a_j\leftarrow x_j-\epsilon$\;
\textbf{real number} $b_j\leftarrow x_j-\epsilon--2^{-j}$\;
\If{$a_j<0$}{
$j\leftarrow 0$\;
$\eta\leftarrow a_j$\;
}
\If{$b_j>0$}{
$j\leftarrow 0$\;
$\eta\leftarrow b_j$\;
}
\Else
{
$j\leftarrow j+1$\;
}
}
\Return $s$
\end{algorithm}
 \begin{itemize}
\item Case 1: $\epsilon>D(y,S_L)$. There exist infinite many $k$ such that $D(A^k\vec{x},y)<\epsilon$. The reason is that the open ball $B(y,\epsilon)$ contains one limit point of the orbit, it must contains infinite many point of the original orbit.

\item Case 2: $\epsilon<D(y,S_L)$. There is a constant lower bound of $D(y,S_L)-\epsilon$ by obtain $b_j=\eta>0$ for some $j$.

Note that $S_L$ is a compact set, we know that for any $k$, we can choose $\vec{q}\in S_L$ such that $D(A^k\vec{x},\vec{q})=D(A^k\vec{x},S_L)$, then
\[D(\vec{y}, A^k\vec{x})\geq D(\vec{y},\vec{q})-D(A^k\vec{x},\vec{q})\geq D(y,S_L)-D(A^k\vec{x},S_L).\]

If $k$ is some integer that we want, $i,e.,$ $D(\vec{y}, A^k\vec{x})\leq \epsilon$, we know that
\[D(A^k\vec{x},S_L)\geq D(y,S_L)-D(y, A^k\vec{x})>D(y,S_L)-\epsilon>\eta>0.\]
According to the proof of Lemma \ref{lem:limitpoint}, we can conclude that there exists a nonsingular $Q$ and a vector $\vec{r}$ such that
$A^k\vec{x}=QG^k\vec{r}$ and $\vec{q}=QC\vec{r}$ holds for any $\vec{q}\in S_L$, where
\begin{eqnarray*}
G=\left(
\begin{array}{cc}
V & 0  \\
 0 & R \\
\end{array}
\right), and~~C=\left(
\begin{array}{cc}
D & 0  \\
 0 & 0 \\
\end{array}
\right),
\end{eqnarray*}
with $D$ and $V$ being diagonal unitary matrices and $R=\oplus_{j} J_{k_j}(\lambda_j)$ with $|\lambda_j|<1$.
Due to Lemma \ref{lem:uniatry} and Lemma \ref{lem:limitpoint}, we can choose $D=V^k$.
Thus,
\[D(A^k\vec{x},S_L)\leq D(A^k\vec{x},\vec{q})=||Q (G^k-C)\vec{r}||\leq s||(G^k-C)\vec{r}||= s||R^k\vec{r'}||,\]
where $s=||Q||>0$ and $\vec{r'}$ is some vector in the smaller dimensional space.
Notice that there exists $0<\lambda<1$ and integer $K$ such that for $k>K$, we have
\[||R^k\vec{r'}||<\lambda^k.\]
Therefore,
\[\eta<D(A^k\vec{x},S_L)<s\lambda^k\Rightarrow k\leq -\log_{\lambda}\frac{\eta}{s}.\]
The rest is to check whether $D(\vec{y}, A^k\vec{x})\leq \epsilon$ holds for $k\leq \max\{-\log_{\lambda}\frac{\eta}{s},K\}$.
 \end{itemize}
 The proof of this theorem is completed.
\qed
\end{proof}
\section{Conclusion}
In this paper, we revisit the orbit problem and study the decidability of so called ``approximate orbit problem". We first provide a complete characterization of the limit points of the given orbit, then we demonstrate a method which gives a good approximation of the distance between the target vector and the set of limit points. The assumption that the radius does not equal to the distance between the target vector and the set of limit points plays a central role in our arguments. If this assumption is valid, the ``approximate orbit problem" can be solved. Without this assumption, this problem seems quite hard as the distance (precise value) is difficult to obtain even the given vector norm is $l^2$ norm.

By removing the assumption with another one, we have an interesting problem: For a given rational orbit $\mathcal{S}(A, \vec{x})$ with $S_L$ being the set of limit points a rational target vector $\vec{y}$, whether there is $k\in\mathbb{N}$ such that $||A^k \vec{x}-\vec{y}||<D(S_L,\vec{y})$.  That is the radius equals to $D(S_L,\vec{y})$, the distance between the target vector and the set of limit points.

\paragraph{Acknowledgement.} We thank S. Ying for his careful reading of the previous version of this paper. N. Yu was
indebted to Professor M. Ying for his constant support during this project. This work was partly supported by the ERC Advanced Grant VERIWARE and the National Natural Science Foundation of China (Grant Nos. 61179030 and 60621062).

\bibliographystyle{splncs03}

\end{document}